\documentclass[a4paper,onecolumn,superscriptaddress,10.5pt,dvipsnames,svgnames,noarxiv]{quantumarticle}
\usepackage{graphicx} 
\usepackage[utf8]{inputenc}
\usepackage{amsmath}
\usepackage{amssymb}
\usepackage{braket}
\usepackage{graphicx}
\usepackage{float}
\usepackage{tikz}
\usepackage{cite}
\usepackage{amsthm}
\usepackage{amssymb}
\newtheorem{definition}{Definition}
\usepackage{xcolor}
\usepackage{amsmath}
\usepackage{pifont}
\usepackage{amsthm}
\usepackage{qcircuit}
\usepackage{cancel}
\usepackage{changes,subcaption,comment}
\usepackage{paralist}
\usepackage{hyperref}
\usepackage{cleveref}
\usepackage[most]{tcolorbox}
\usepackage{booktabs}
\usepackage{tabularx}
\usepackage{array}

\newtcolorbox{examplebox}{
    colback=magenta!5,
    colframe=magenta!50!white,
    coltitle=white,
    colbacktitle=magenta!50!white,
    fonttitle=\bfseries,
    title=Example,
    boxrule=0.8pt,
    arc=2mm,
    left=2mm,
    right=2mm,
    top=2mm,
    bottom=2mm
}

\usepackage{enumitem}
\setlist[itemize]{noitemsep}
\setlist[enumerate]{noitemsep}

\newtheorem{example}{Example}

\newcommand{\ba}{\begin{eqnarray}}
\newcommand{\ea}{\end{eqnarray}}
\newcommand{\C}[1]{\mathcal{#1}}
\newcommand{\Rprec}[1]{\xrightarrow{#1}}

\title{Causal Trade-Offs in Superluminal Extensions of Relativity. A no-go theorem}

\author{Amrapali Sen$^\dagger$}
 \email{amrapalisen.phy@gmail.com}\thanks{Corresponding author}
\affiliation{ICTQT, University of Gda\'{n}sk, Prof. Marii Janion 4, 80-309  Gda\'{n}sk, Poland}

\author{Matthias Salzger$^\dagger$}
 \email{ matthias.salzger@outlook.com}
\affiliation{ICTQT, University of Gda\'{n}sk, Prof. Marii Janion 4, 80-309  Gda\'{n}sk, Poland}

\author{V. Vilasini}
 \email{vilasini@inria.fr}
\affiliation{Univ. Grenoble Alpes, Inria, 38000 Grenoble, France}

\begin{document}
\maketitle
\section{Abstract}

Extensions of relativity that include superluminal reference frames have recently been proposed by Dragan and Ekert as a way to reconcile relativistic symmetry with quantum uncertainty. In this work, we develop a causal modelling framework to formally study this proposal and prove a no-go theorem demonstrating that a theory permitting superluminal reference frames can avoid superluminal signalling only if it severely constrains which spacetime regions are accessible to operational protocols, if the choice of reference frame has observable consequences or if free variables do not exist or their existence depends on the reference frame. Crucially, our result goes beyond the familiar intuition that superluminal transformations may lead to superluminal signalling; our no-go result shows that Dragan and Ekert's proposal of introducing objective uncertainty does not, by itself, eliminate this conflict. In fact, because superluminal signalling might instead be avoided by relaxing other fundamental assumptions in the
presence of superluminal reference frames, the proposed objective uncertainty may not be necessary at all. Our results clarify the limitations of uncertainty-based resolutions and the causal trade-offs faced by superluminal extensions of relativity.

\def\thefootnote{$\dagger$}\footnotetext{These authors contributed equally to this work}\def\thefootnote{\arabic{footnote}}

\section{Introduction}
Relativistic causality constrains the ordering of events through the light-cone structure of Minkowski spacetime. Standard special relativity is conventionally assumed to prohibit superluminal reference frames, without any assumption(s) explicitly ruling it out. Andrzej Dragan and Artur Ekert have proposed extending the relativity principle to superluminal observers, arguing that quantum indeterminacy may arise as a consistency requirement of such an extension \cite{dragan2020quantum}. Their central claim is that local-deterministic descriptions become inconsistent under superluminal boosts, and that objective uncertainty is therefore required. This raises several questions. What does a locally-deterministic description mean in this context? Are we really talking about both local and deterministic aspects of superluminal particles? What precisely is not determined: the occurrence of the event itself, its spacetime localisation, or the outcome associated with it? How large must the non-local or non-deterministic or non-(local-deterministic) regions? Do these regions prevent superluminal signalling? Must any other assumptions be relinquished for this proposal to succeed? Moreover, is such uncertainty genuinely necessary, or could superluminal signalling instead be avoided by relaxing those assumptions? These questions motivate the analysis developed below.\\

Our aim here is not to determine whether indeterminism follows from superluminal transformations (it was discussed in another paper by AS where it is argued that even if any indeterminacy is necessary, it is different from standard quantum indeterminacy  \cite{Sen2026Indeterminism}). 
In this paper, rather than focusing on (any) indeterminacy \cite{Sen2026Indeterminism} per se, and without imposing determinacy as an assumption, we revisit this idea from the perspective of causal modelling in order to shed light on the above questions through a formal framework while deconstructing the underlying assumptions. We formulate the relevant operational and causal assumptions explicitly and examine which combinations are compatible with superluminal reference frames. We examine the embedding of causal structures under superluminal transformations (hereafter denoted by SpTs).\\ 

We then prove a no-go theorem for theories admitting SpTs. Rather than relying on a particular interpretation, our result is formulated as an incompatibility between a set of independently motivated assumptions. The central tension can be summarised as follows. Correlations that support ordinary communication in one frame may be mapped by a superluminal transformation to correlations between spacelike-separated operational regions in another. Avoiding superluminal signalling then requires at least one additional assumption about what kinds of spacetime regions are operationally accessible, frame-independence of experimental observations, or free choice to fail.

In a conventional subluminal context, each of these assumptions (apart from SpTs) is natural in existing operational, relativistic, causal-modelling frameworks \cite{vilasini2022general,Pearl2009,OreshkovCostaBrukner2012, barrett2020, henson2014theory}. The theorem does not by itself single out a unique resolution, but instead characterises the trade-offs that any theory incorporating superluminal reference frames must confront. Additionally, we discuss what these trade-offs imply for the argument of \cite{dragan2020quantum}, finding that they either obviate the need for uncertainty to avoid superluminal signalling or require uncertainty many orders of magnitude larger than what is observed in the experimental world. \\ 

 The paper is organised as follows. In \cref{sec3}, we introduce superluminal Lorentz transformations. In \cref{sec5}, we introduce a communication protocol that forms the operational basis of our analysis and review the argument of Dragan and Ekert regarding superluminal reference frames and the emergence of objective uncertainty \cite{dragan2020quantum}. We analyse their reasoning in the context of a simple experimentally observed communication protocol between two agents and discuss how the argument applies to different physical realizations of the protocol. Next in \cref{sec4}, we lay down the framework necessary to discuss causal modelling in the presence of SpTs. In the framework, we generalise the affects framework of \cite{vilasini2022general} to allow the spacetime causal structure to depend on the reference frame, reflecting the fact that light-cone relations need not be preserved under SpTs. We also generalise the framework by allowing observed variables to be embedded into extended spacetime regions, rather than being associated with individual spacetime points, following the approach of \cite{Vilasini2024Embedding}. In \cref{sec6} we formalise the causal structure of the communication protocol and state the assumptions underlying our framework. Using these assumptions, we present our no-go theorem. Finally, in \cref{sec7} we discuss the consequences of relaxing each assumption, and how it relates to the argument presented in \cite{dragan2020quantum}. Then we conclude in \cref{sec7} and leave some discussion in \cref{sec8}.  

\section{Superluminal Lorentz transformations}\label{sec3}

We briefly review the derivation of Lorentz transformations \cite{parker1969} and indicate how the analysis extends to superluminal frames. For simplicity, we restrict attention only to the the $(1+1)$-dimensional case.\\

Consider two inertial reference frames $(t,x)$ and $(t',x')$, where the primed frame moves with constant velocity $V$ relative to the unprimed one. The principle of relativity implies that the transformation between these coordinate systems must be linear, since inertial motion in one frame should correspond to inertial motion in the other. Moreover, the coefficients of the transformation may depend only on the relative velocity $V$. Under these assumptions, the most general linear transformation relating the two
frames can be written as
\begin{equation}\label{eq1}
x' = A(V)x + B(V)t ,
\end{equation}
where the functions $A(V)$ and $B(V)$ depend solely on the relative velocity.
Consistency of the relativity principle further requires that the inverse
transformation has the same functional form but with the velocity reversed.
Hence one must also have
\begin{equation}
x = A(-V)x' + B(-V)t'.
\end{equation}

We can then pick an origin \(x'=0\) of the primed frame that follows the trajectory \(x=Vt\), which in turn constrains the coefficients in \cref{eq1} through \(B(V)=-VA(V)\). Consistency under a reversal of the spatial orientation, which also maps \(V\mapsto -V\), allows two possibilities: \(A(-V)=A(V)\) and \(A(-V)=-A(V)\). 
By imposing the principle of relativity and hence requiring successive transformations between inertial frames to retain the same functional form and imposing the existence of a
universal limiting speed $c$, one obtains the familiar Lorentz transformations for $|V|<c$ for the symmetric choice. The antisymmetric choice with  $|V|>c$ leads to the so-called superluminal Lorentz transformations, which take the following form:

\begin{equation}\label{LB2}
    x'=\pm\frac{V}{|V|}\frac{x-Vt}{\sqrt{V^2-1}}\\\hspace{1cm}
    t'=\pm\frac{V}{|V|}\frac{t-Vx}{\sqrt{V^2-1}}.
\end{equation}

The overall sign in these expressions cannot be fixed uniquely, since there is no limit $V \rightarrow 0$. Consequently, the choice of sign is conventional, and we adopt the positive branch whenever required.

\medskip
\noindent\textit{Note.} In the limit of infinite relative velocity,
often referred to as a \emph{superboost} \cite{Dragan2023}, the transformation formally
interchanges the temporal and spatial coordinates, such that $t'$
becomes proportional to $x$ and $x'$ to $t$.
\section{Communication protocol and Dragan and Ekert argument}\label{sec5}

To analyse the implications of superluminal reference-frame transformations for causal structure, we begin by considering a simple communication protocol between two agents, depicted in \cref{cp}.\footnote{The particular physical realization described below is not an additional assumption of our no-go theorem, nor is the theorem restricted to this protocol. Rather, it provides an intuitive example for our no-go theorem.} Consider the following set-up: Alice wishes to communicate a uniformly distributed bit $X$ to Bob whose guess for the bit we denote with $B$.
We assume the following strategy: Alice has some device which she can switch on and off. When the device is on, it can emit a particle. If it is off, it cannot. We can, for example, imagine that the device consists of a radioactive source encased by very thick lead walls such that all radiation is stopped. A gate in the walls can be opened (and closed) at the press of a button and an emitted particle can escape the device while the gate is open. If $X=0$, Alice leaves the device off. If $X=1$, Alice operates the device within a laboratory region of spatial extent $\Delta x_A$ for a time interval $\Delta t_A$. Meanwhile, Bob searches for the particle within a laboratory region of spatial extent $\Delta x_B$ for a time interval $\Delta t_B$, assigning $B=1$ if he does, and $B=0$ if he does not. The protocol need not be perfect: Bob may fail to detect an emitted particle or may register an unrelated particle.  
Clearly, such or similar devices exist in the real world and as long as Alice is in the causal past of Bob, she can communicate to him, i.e. $p(B = X) > 0.5$, since $X$ is uniformly distributed.\\

\begin{figure*}[htp]
		\centering
		\includegraphics[width= 0.9 \textwidth]{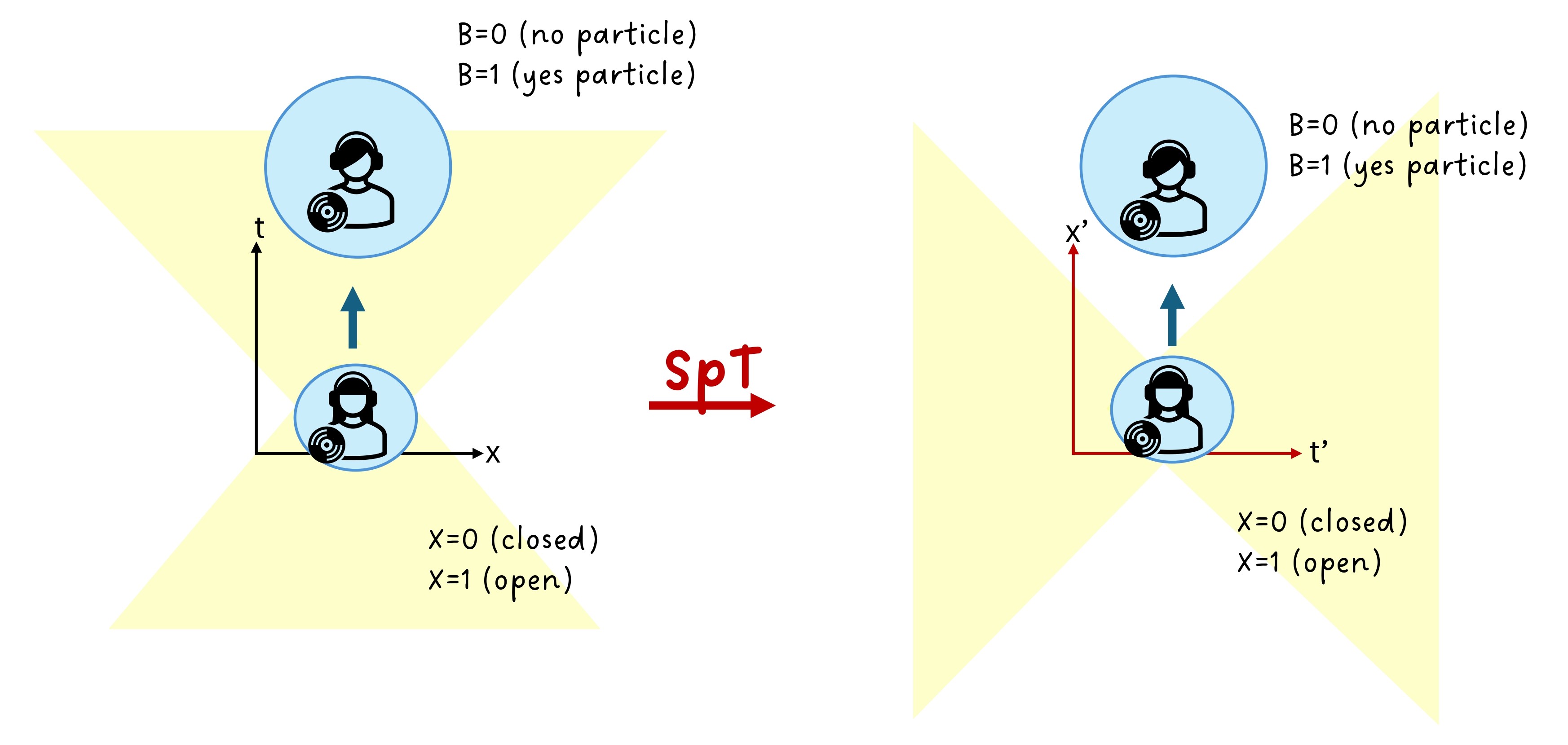}	\caption{Communication protocol. The coordinates in black represent our frame, while the coordinates in red represent a superluminally boosted reference frame (of let's say Charlie) with respect to us.}
		\label{cp}        
	\end{figure*}

However, if one includes a superluminal observer, Charlie the situation changes. In the limit of Charlie moving at infinite velocity, the roles of the time coordinate and the space coordinate switch according to \cref{LB2}. This also means that the previously timelike path between Alice and Bob becomes a spacelike path and it seems like superluminal signalling is possible from Alice to Bob in Charlie's reference frame. One could now invoke the possibility of fundamental spatiotemporal uncertainty, but the physical implementation constrains both operations: Alice can emit the particle only within a laboratory region of spatial extent $\Delta x_A$ and during a time interval $\Delta t_A$, while Bob can search for it only within a laboratory region of spatial extent $\Delta x_B$ and during a time interval $\Delta t_B$. These constraints bound the possible spacetime locations of the corresponding input and output events. If these laboratory regions are sufficiently small relative to the spacetime separation between Alice and Bob, then the two regions remain spacelike separated in Charlie’s reference frame, and the communication appears superluminal, as illustrated in \cref{cp}.\\

To analyse how this communication protocol compares with arguments in \cite{dragan2020quantum}, we illustrate the argument of Dragan and Ekert in \cref{ED}. On the left, they consider a superluminal particle emitted at an event $A$ and absorbed at an event $B$. They call the description locally deterministic if information contained locally in the source’s worldline prior to $A$ determines the exact time of emission. Under a superluminal transformation, on the right, however, the roles of the emission and absorption events may be reversed. The event $A$ whose occurrence was fixed by the source’s past in the left frame can then appear, in another frame, to depend on an event $B$ which is now the source. They conclude that a frame-independent description admitting superluminal observers cannot retain a locally deterministic account and must instead assign objective indeterminacy to the emission event.\\

It is important, however, to distinguish the locality and determinism components of this conclusion. In their discussion of the scenario depicted in \cref{ED} in \cite{dragan2020quantum}, Dragan and Ekert  postulate locality by assuming that the ``source particle prior to the event A contains locally all the information necessary to predict the exact moment of emission of a superluminal particle at $A$ ". They subsequently argue that local-\emph{determinism} fails because ``the past world-line of the particle $B$ carries no information ... the emission at $B$ was completely spontaneous and fundamentally unpredictable". The direct conclusion of this argument is that the relevant information cannot remain locally encoded in the past worldline of the source particle in every reference frame. In this sense, the argument establishes a failure of the \emph{locality} requirement imposed on the explanation. It does \emph{not} however provide a separate argument against \emph{determinism}. For example, their argument by itself does not exclude nonlocal-deterministic explanations. 
Their argument indeed shows that superluminal extensions are incompatible with the particular notion of locality adopted in \cite{dragan2020quantum}, rather than independently demonstrating that determinism must also be abandoned. In addition, they do not explicitly analyse no-superluminal signalling for tachyons, but remark in passing that ``If we had a source of superluminal particles at our disposal, we would not be able to use it to send any information". This statement appears to assume that agents lack sufficient control over the relevant emission and absorption events. However, such a restriction on operational control is neither formulated precisely nor shown to follow from the abandonment of \emph{local-determinism} or from the principle of relativity. \\

Furthermore, based on their apparent violation of local determinism they suggest there is fundamental uncertainty. However, this is not established formally via a rigorous framework and definitions of concepts such as local determinism, uncertainty, or signalling \cite{dragan2020quantum}. The uncertainty they refer to is in the space-time localisation of the emission or absorption of the particle, which means that the corresponding event cannot be associated with a sharply localised spacetime point and must instead be represented by an extended spacetime region. This motivates our formal treatment of space-time uncertainty. It is represented by the extended blue regions in our \cref{cp}, \cref{ED}, \cref{swap2} in this section.
We refer to this space-time uncertainty as D-E uncertainty. The notions of locality and determinism are familiar from discussions of Bell’s theorem, where spacetime based distinctions of several physical assumptions are often revisited\cite{RevModPhys.92.021002}. However, space-time uncertainties are not necessary in the context of Bell's theorem since a Bell inequality can be violated using entangled spins, which could in principle be perfectly localised in spacetime. Hence, D-E uncertainty is very different from the local-determinism assumption when it comes to Bell's theorem, but similar in sense to local-determinism in \cite{dragan2020quantum}.\\

The argument of Dragan and Ekert refers informally to agents having access to the histories of the particles at $A$ and $B$, but it does not formulate those access conditions as an operational protocol. We therefore introduce explicit agents and observed variables and analyse the resulting communication protocol within a causal-modelling framework.

\begin{figure*}[htp]
		\centering
		\includegraphics[width= 0.8 \textwidth]{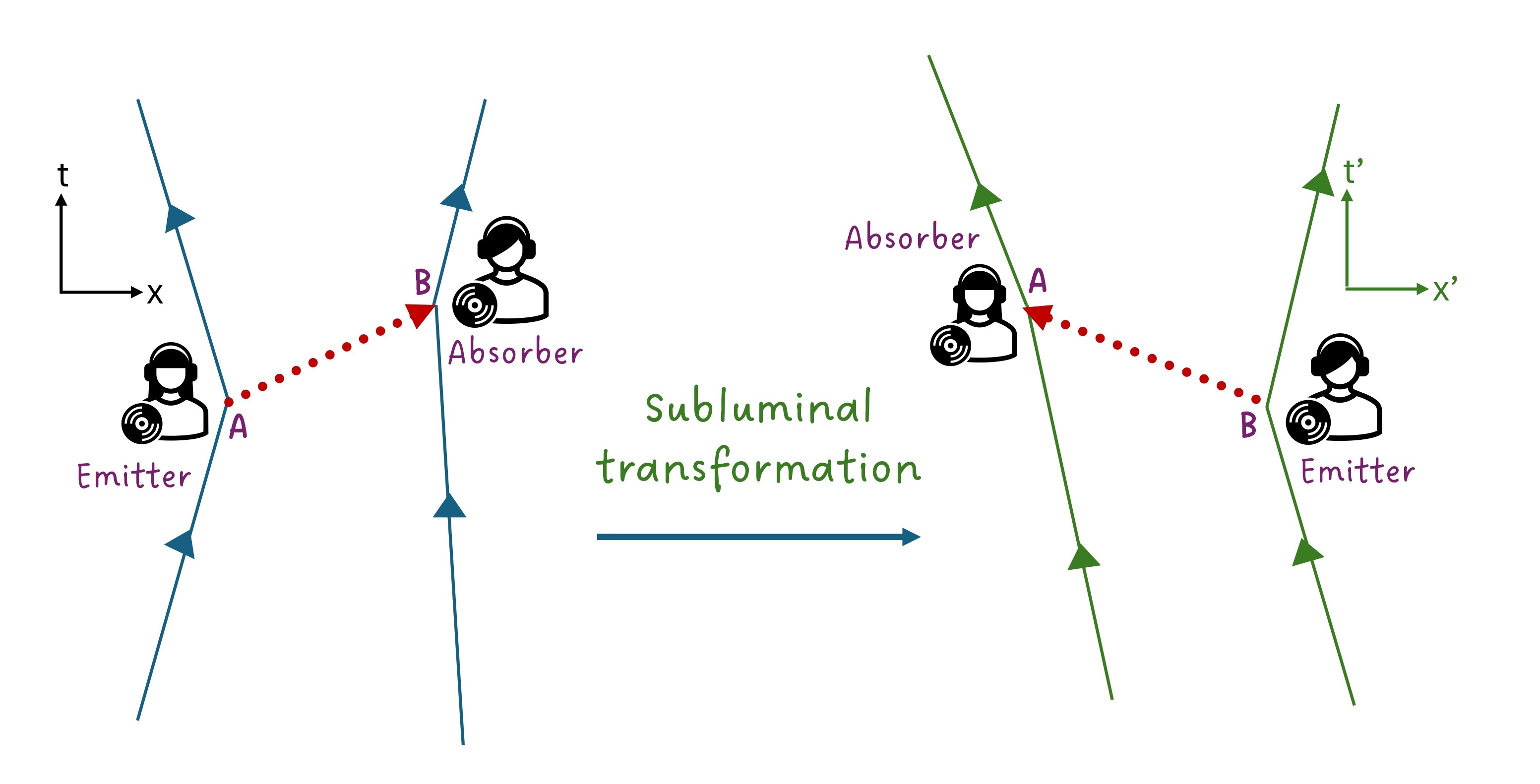}	\caption{A tachyon is sent from Alice to Bob. Dragan and Ekert \cite{dragan2020quantum} argue that for this to be consistent with relativity, no local-deterministic model of the tachyon's emission is available on the absorber's, i.e., Bob's side. Under a subluminal Lorentz boost, the roles of emitter and absorber change. Thus, by (relativistic) symmetry, they argue there can be no local-deterministic model on Alice's side either.}
		\label{ED}        
	\end{figure*}

We now apply the reasoning of Dragan and Ekert to the communication protocol in \cref{cp} (note that the fact that our protocol considers what happens under a superluminal transformation while the argument depicted in \cref{ED} considers a subluminal transformation does not matter here. Indeed, a subluminal particle under a superluminal transformation is analogous to a superluminal particle under a subluminal transformation.)\\

We now distinguish the two claims that have been made in their argument. First, claim \textbf{A}, superluminal reference-frame transformations can reverse the apparent roles of the sender and receiver. Second, claim \textbf{B}, is that the absence of operational control over the precise spacetime locations of the relevant events introduces D-E uncertainty (which they claim to be quantum) and thereby prevents superluminal communication. However, even after applying \textbf{A} and \textbf{B}, to our communication protocol in \cref{swap2}, we can see superluminal signalling (here signalling between spacelike regions is unavoidable).\\

However, Claims \textbf{A} and \textbf{B} prevent superluminal signalling only if these regions are sufficiently extended to preserve the required causal ordering after transformation. The existence of uncertainty alone is therefore insufficient; its required extent must also be specified.\\
\begin{figure*}[htp]
		\centering
		\includegraphics[width= 0.9 \textwidth]{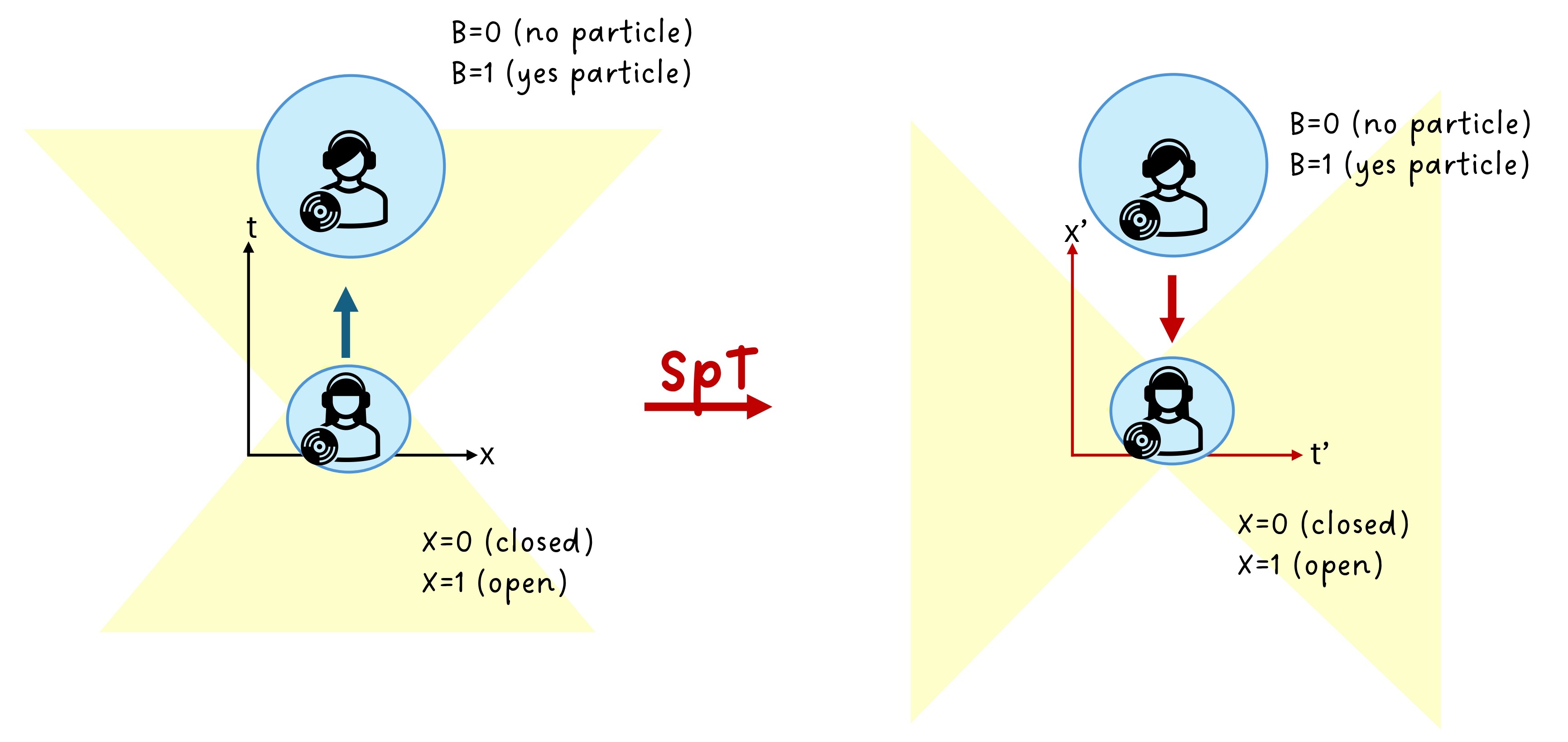}	\caption{Claim \textbf{A} and claim \textbf{B} applied to our communication protocol. The coordinates in black represent our frame, while the coordinates in red represent a superluminally boosted reference frame (of let's say Charlie) with respect to us.}
		\label{swap2}        
	\end{figure*}

The case of communication mediated by photons provides a useful contrast. Because lightlike trajectories remain lightlike under the transformations considered here, their ordering may reverse without the corresponding events becoming spacelike separated. In this case, claim \textbf{A} can exchange the roles of sender and receiver without requiring claim \textbf{B}: D-E uncertainty to prevent superluminal signalling. By contrast, for communication mediated by massive particles, the transformed events may become spacelike separated, and the proposed resolution must invoke claim \textbf{B}. Thus, the argument does not require D-E uncertainty uniformly across different types of information carriers. 
This raises several questions about the interpretation and scope of the proposed uncertainty. Since D-E uncertainty, as they propose and we discuss above, is about relaxing locality, how does then one conclude claims on determinism when any assumption about determinism was never made? How, then, can any analogy with quantum theory be justified? Even if one agrees that some D-E uncertainty (on locality assumption) is required, how to find a measure of such uncertainty and what does it really depend on? If some form of uncertainty is indeed unavoidable, how should we quantify them, given the size and mass of the information carrier? More fundamentally, is D-E uncertainty alone sufficient to ensure consistency, or must other operational principles also be abandoned? What consequences does their argument have for operational protocols where the emission of a particle is actively controlled by an agents' physical device? Such questions are left open to the readers of \cite{dragan2020quantum}. We explore these formally.

\section{Framework}\label{sec4}

The assumptions underlying Dragan-Ekert's proposal and the conditions required for its consistency have not yet been formulated within a sufficiently explicit framework to permit a systematic assessment of the challenges it may face. In this section we formalise and distinguish two independent formalisms used to study causal correlations. First, we review information-theoretic causal models, which are defined independently of spacetime. Second, we introduce the spacetime structure which is associated with and may depend on the reference frame. Third, we connect these structures by embedding operational variables into frame-dependent spacetime regions. Keeping these components distinct allows us to identify which assumptions concern causal relations, spacetime ordering, or the compatibility between them. 

\subsection{Information-theoretic causal models}

Taking an operational approach to physical theories, it is possible to define causal influence in information-theoretic terms, capturing a flow of information
between observed and unobserved systems. There are many different approaches for doing so \cite{vilasini2022general, Pearl2009,OreshkovCostaBrukner2012, barrett2020, henson2014theory}. Here we adopt and extend the approach of \cite{vilasini2022general}, where we formulate this information-theoretic notion in a minimal and theory-independent manner while disentangling it from spacetime and reference frames. This will allow us to clearly separate the information-theoretic assumptions
of our no-go theorem from the spatiotemporal and reference-frame based ones.

\begin{definition}[Causal graph]

A causal graph $\C{G}$ is a directed graph
\[
\C{G} = (V_{\C{X}}, E_{\C{X}}),
\]
whose vertex set $V_{\C{X}}$ is partitioned into observed and unobserved nodes. The set of observed nodes, $V_{\C{X}_{\mathrm{obs}}}$, corresponds to the observed variables $\C{X}_{\mathrm{obs}} \in \C{X}$, with each node representing a classical random variable of the same name. The remaining nodes, denoted by $\C{X}_{\mathrm{nobs}}$, represent unobserved systems, which may be classical, quantum, or post-quantum. The directed edges $E_{\C{X}}$ encode the causal relationships among the nodes.

\end{definition}

Given a causal graph $\C{G}$, the parents of a node correspond to the set of all nodes (observed or unobserved) that have a directed edge to that node. Moreover for any two nodes $X$ and $Y$, a directed path from $X$ to $Y$, denoted by $X\rightarrow...\rightarrow Y$ refers to a path from $X$ to $Y$ where all arrows in the path point in the same direction, away from $X$ and towards $Y$.

\begin{definition}[Free variable]
\label{def: freevar}
Given a causal causal graph $\C{G}$, an observed variable
$X\in \C{X}_{\rm obs}$
is called free if it has no parents in $\C{G}$.
\end{definition}

In other words, \textit{free} variables have no incoming arrows in a causal graph. The notion of free variable used here (similar to that of \cite{vilasini2022general}, by one of us (V.V)) is entirely graph-theoretic and with respect to the information-theoretic causal graph and therefore independent of any notion of spacetime. We also note that, in \cite{colbeck2013short}, the free variables are defined relative to the lightcone structure of a spacetime in which they are embedded, which is appropriate and meaningful when imposing constraints such as no-superluminal causation and frame-independence of the lightcone structure. While these constraints are well motivated in physical scenarios involving subluminal reference frames and observers, they are not essential for our purposes. In particular, our no-go theorem for SpTs requires the more general and weaker notion of free variables introduced above, which is formulated independently of relativistic causality principles. However, when such constraints are imposed, the causal graph's arrows align with the spacetime's lightcone structure in all frames and the above notion recovers the more well known notion of Colbeck-Renner \cite{colbeck2013short}.\\

\begin{definition}[Weak Causal model]
\label{def: causalmodel}
A \emph{weak causal model} (which we will simply call causal model in this paper) consists of a causal graph $\C{G}$ together with a joint probability distribution $P_{\C{G},\C{X}_{\mathrm{obs}}}$ over the set of observed random variables $\C{X}_{\mathrm{obs}}$ of $\C{G}$ that satisfy the following linking property: for any $X, Y\in \C{X}_{\mathrm{obs}}$ such that $X$ is a free variable,
$$\not\exists \text{ directed path } X\rightarrow...\rightarrow Y \text{ in } \C{G} \quad \Rightarrow \quad P_{\C{G},\C{X}_{\mathrm{obs}}}(X,Y)=P_{\C{G},\C{X}_{\mathrm{obs}}}(X)P_{\C{G},\C{X}_{\mathrm{obs}}}(Y).$$

\end{definition}

The above definition of causal model is inspired by \cite{vilasini2022general} and also permits cyclic causal graphs.
We have deliberately not assigned any meaning to the edges, to keep it minimal and theory independent. This means that the definition applies to all existing acyclic classical, quantum and post-quantum causal models and several cyclic ones across these theories too (see below), which give physical meaning to the edges in terms of causal mechanisms that correspond to information-theoretic channels of the theory. \\

We note that our definition above is strictly more general than the previous definition of \cite{vilasini2022general}, as the linking property we have used is weaker. The main idea of the linking property is to constrain the observed distribution according to the causal graph, particularly, separations between observed nodes in the graph imply (conditional) independences or absence of correlations in the corresponding observed distribution. In \cite{vilasini2022general}, the linking property is given by $d$-separation, which is satisfied in general acyclic classical, quantum and post-quantum causal models \cite{pearl2009causality, henson2014theory,barrett2020}, several cyclic models \cite{Bongers_2021,vilasini2022general}, but is known to fail in certain cyclic models \cite{Neal_2000,ferradini2025}. The linking property used here is implied by the $d$-separation linking property, but it is weaker as it is only applied when $X$ is parentless and considers individual variables $X$ and $Y$, as opposed to $d$-separation which considers three disjoint sets of observed nodes $X$, $Y$ and $Z$ and the conditional distribution $P_{\C{G},\C{X}_{\mathrm{obs}}}(X,Y\mid Z)$. In the rest of the paper, whenever we say causal model (without explicitly saying ``weak''), it should be understood in the sense of the general definition above, and our no-go theorem will apply to all such models.

We then define the notion of signalling used in this work formulated within this framework. This way of defining signalling is however quite standard from Bell type scenarios and other quantum information protocols and frameworks\cite{Brunner2014}.
\begin{definition}[Signalling using a free variable]
\label{def: signalling}
Consider a causal model on a causal graph $\C{G}$ associated with the observed distribution $P_{\C{G},\C{X}_{\mathrm{obs}}}$. Given a free variable $X\in \C{X}_{\mathrm{obs}}$ and any other observed variable $Y\in \C{X}_{\mathrm{obs}}$, we say that $X$ \emph{signals to} $Y$ whenever $X$ and $Y$ are correlated according to $P_{\C{G},\C{X}_{\mathrm{obs}}}$ i.e.,
$$P_{\C{G},\C{X}_{\mathrm{obs}}}(X,Y)\neq P_{\C{G},\C{X}_{\mathrm{obs}}}(X)P_{\C{G},\C{X}_{\mathrm{obs}}}(Y).$$
\end{definition}

Note that if $X$ is not free, correlation with $Y$ does not generally imply signalling or causal influence from $X$ to $Y$ (hence the famous slogan, correlation is not causation). However, when $X$ is free, the causal model guarantees that it can be correlated to $Y$ only if it has a directed path to $Y$, here the correlation captures that freely choosing and fixing $X$ will lead to a change in the distribution on $Y$ which is precisely signalling. This then allows us to operationally infer the presence of causal influence (through the directed path) of $X$ on $Y$. Note that at this point, we have not defined \emph{superluminal} signalling, as this is a concept that involves both the information-theoretic notion (i.e., signalling) and spatiotemporal ones (i.e., future lightcone), and will be defined later, after the latter is introduced.

\subsection{Reference frames and spacetime}\label{RF}

We now introduce reference frames and spacetime independently of causal models. We first specify how spacetime will be represented. The results of this work depend only on the light-cone ordering of spacetime events, rather than on metric distances or other geometric properties of spacetime. Hence, it is sufficient to model spacetime (events) in terms of partially ordered sets (posets).  
 In standard special relativity, where only subluminal transformations are considered, the light-cone structure is invariant under Lorentz transformations and hence frame-independent. As a consequence, all inertial observers agree on the ordering of events induced by the lightcones, and a single partially ordered set suffices to represent spacetime (e.g., causal boxes \cite{Portmann2017CausalBoxes}, process matrices \cite{OreshkovCostaBrukner2012}, or the affects framework \cite{vilasini2022general} by one of us (V.V)). In contrast, once superluminal reference-frame transformations are admitted, this invariance may no longer hold. As we will see later, such transformations may fail to preserve the light-cone ordering of spacetime events, in the sense that whether one event lies in the future light-cone of another can depend on the reference frame under consideration. To model this generally, we will define a separate poset $\C{T}_\C{F}$ to model the spacetime lightcone structure relative to each reference frame $\C{F}$.\\ 

Here, we are considering reference frames specifically as frames for spacetime, so that each frame comes with a causal partial order induced by its light-cone structure. This is the structure relevant for the present analysis and justifies representing each spacetime frame by a poset. More generally, in quantum reference-frame frameworks, one can also consider reference frames associated with other physical degrees of freedom, such as spin or direction, which need not carry any natural partial order. Such reference frames are outside the scope of the present framework, where we use “reference frame” specifically in the spacetime sense.

\begin{definition}[Reference frame and partial order]\label{def:reference-frame}
We model spacetime as a set of events $\C{T}$\footnote{This is essentially analogous to how one models spacetime as a manifold in general relativity. A manifold is a topological space that locally resembles Euclidean space. Ignoring this additional structure, which we are not interested in, spacetime simply becomes a set.}. A reference frame is defined by a label $\C{F}$ together with a corresponding partially ordered set
\[
 (\C{T}_\C{F}, \prec_\C{F}),
\]
where $\C{T}_\C{F}$ is isomorphic to $\C{T}$ and $\prec_\C{F}$ is the spacetime partial order induced by the light-cone structure in reference frame $\C{F}$.\end{definition}

We emphasise that $\prec_\C{F}$ denotes the order specified by the lightcone structure, which in our framework can depend on the chosen reference frame. 
Further, we note that our definition of reference frames is abstract and general, also, to be as general as possible for our no-go theorem. 
It is essentially just a label $\C{F}$ with an associated partial order with no further structure. Thus, our results apply to any framework of reference frames compatible with the idea of a background spacetime. For example, it is irrelevant whether 
the reference frames can be associated with physical systems (as is the case in the framework of quantum reference frames \cite{Ahmad2022QuantumRelativitySubsystems, Hoehn2023SubsystemRelativity, CastroRuiz2025RelativeSubsystems, Cepollaro2025SubsystemCoherence, sen2026superluminal}) 
or not. 

\begin{definition}[Reference-frame transformation]\label{RF}
A reference-frame transformation is a map
\[
\C{S}_{\C{F} \to \C{F'}} : \C{T}_\C{F} \rightarrow \C{T}_\C{F'}.
\]
\end{definition}
At this stage, we impose no further structural assumptions on the reference-frame transformations. In particular, it is not required to preserve the partial order (like in standard subluminal Lorentz transformations which forms a restricted subclass of our maps), nor is any regularity condition such as continuity or differentiability (this would require additional structure on the partial orders) assumed. This level of generality is deliberate, as it allows the framework to encompass both standard Lorentz transformations as well as more general transformations, including superluminal reference-frame transformations, within a unified abstract setting.\\

\begin{example}[Superluminal Lorentz transformations]
As an illustrative example, let us model the superluminal Lorentz transformations in 1+1 dimensions from \cref{LB2} in this framework. To each frame, we associate a partial order $(\C{T}_\C{F}, \prec_\C{F})$ capturing the spacetime and lightcone structure from the perspective of this frame $\C{F}$. One may take $\C{T}_\C{F}$ to be a coordinate representation of spacetime
in frame $\C{F}$, with $\prec_\C{F}$ given by the corresponding light-cone order. A subluminal Lorentz transformation from reference frame $\C{F}$ to reference frame $\C{F}'$ is an order-preserving map $\C{S}_{\C{F} \to \C{F'}} : \C{T}_\C{F} \rightarrow \C{T}_\C{F'}$ as it preserves the lightcone structure. On the other hand, a superluminal Lorentz transformation (SpLT) is a map that does not preserve the order. In particular, the superboost $\C{S}_{v\to\infty}$ which exchanges time and space maps and timelike separated events to spacelike separated events and vice versa, i.e., $t \prec t'$ or $t' \prec t \Leftrightarrow \C{S}_{v\to\infty}(t) \not \prec \C{S}_{v\to\infty}(t' )$ or $\C{S}_{v\to\infty}(t') \not \prec \C{S}_{v\to\infty}(t)$ (see \cref{ordernp}).
\end{example}

\begin{definition}[Regions and ordering on regions]
\label{def: regions}
Given a reference frame $\C{F}$ and associated poset $\C{T}_\C{F}$, a spacetime region $R$ relative to $\C{F}$ is a subset of $\C{T}_\C{F}$. The set of regions $\C{R}(\C{T}_\C{F})$ is given by the power set of $\C{T}_\C{F}$. For any two regions $R, R' \in \mathcal{R}(\C{T}_\C{F})$, we define the ordering relation
\begin{equation}
    R \xrightarrow{\C{F}} R'
    \Leftrightarrow
    \exists r \in R, r' \in R' :
    r \prec_{\C{F}} r'.
\end{equation}
\end{definition}

The above definition of spacetime regions is similar to that of \cite{Vilasini2024Embedding}, but made relative to the poset $\C{T}_F$ of a given frame $\C{F}$. We note that the order relation $\xrightarrow{\C{F}}$ is not a partial order as, for example, any two regions $R, R'$ with non-empty intersection satisfy $R \xrightarrow{\C{F}} R'$ and $R' \xrightarrow{\C{F}} R$. In fact, this order relation is not even a preorder in general as it does not have to satisfy transitivity.

\subsection{Causal models in reference frames and spacetime}
Until this stage, no compatibility constraints between the causal structure of a causal model and the structure of spacetime have been imposed. There are therefore two distinct notions of causal ordering in the framework. The first is the information-theoretic causal order, defined by directed paths in the causal model $\mathcal{G}$. The second is the spacetime causal order associated with each frame $\C{F}$, given by the light-cone structure of $(\C{T}_{\mathcal{F}},\prec_{\mathcal{F}})$ which induces the order $\xrightarrow{\mathcal{F}}$ on the (frame-dependent) spacetime regions $\C{R}(\C{T}_\C{F})$.\\

We now proceed to relate these notions of causal order for the case of observed variables (which are the operationally accessible variables) via the concept of embedding and subsequently define a weak compatibility condition between the orders which captures a relevant no superluminal signalling condition in the presence of a spacetime embedded protocol.

\begin{definition}[Embedding of observed variables]
\label{def: embedding}
Consider a set $\C{X}_{\mathrm{obs}}$ of observed random variables. For each reference frame $\C{F}$, an embedding of these variables in the frame poset $\C{T}_\C{F}$ is a map
\[
\C{E}_\C{F} : \C{X}_{\mathrm{obs}} \rightarrow \mathcal{R}(\C{T}_\C{F}),
\]
where $\mathcal{\C{R}}(\C{T}_\C{F})$ denotes the set of regions (equivalently, the power set) of $\C{T}_\C{F}$. In addition, these embedding maps have to be consistent in the sense that for all $\C{F}, \C{F}'$, $\C{E}_{\C{F}'} = \C{S}_{\C{F}\rightarrow \C{F}'} \circ \C{E}_{\C{F}}$, with the action of the reference frame transformation on a region $R \in \C{R}(\C{T}_{\C{F}})$ is $\C{S}_{\C{F}\rightarrow \C{F}'}(R) = \{\C{S}_{\C{F}\rightarrow \C{F}'}(x) | x \in R\}$. 
\end{definition}

 In other words, the region associated with a observed variable in one reference frame automatically determines its region in every other frame related to it by the chosen reference-frame transformation.

\begin{definition}[Spacetime embedded protocol]
\label{def: protocol_spacetime}
 
Given a spacetime $\C{T}$, a spacetime embedded protocol associated with a set $\C{X}_{\mathrm{obs}}$ of observed variables is given by 
\begin{enumerate}
    \item a family of causal models, one for each frame $\C{F}$, where each such causal model is associated with the graph $\C{G}_{\C{F}}$ and observed distribution denoted $P_{\C{F}}$ (for simplicity, we have dropped the graph and variables $\C{X}_{\mathrm{obs}}$ from the  subscript relative to \cref{def: causalmodel}) and,
    \item for each frame $\C{F}$, an embedding $\C{E}_{\C{F}}$ of $\C{X}_{\mathrm{obs}}$ in $\C{T}_{\C{F}}$ as defined in \cref{def: embedding}.
\end{enumerate}
\end{definition}

The above definitions ensure that the observed variables $\C{X}_{\mathrm{obs}}$ of any operational protocol are identified in a frame-independent manner, while allowing for the causal model (both the causal graph $\C{G}_{\C{F}}$ and observed distribution $P_{\C{F}}$) to be frame-dependent, as well as the spacetime causal order $\C{T}_{\C{F}}$ and embedding $\C{E}_{\C{F}}$ of the observed variables to be frame-dependent up to respecting the consistency condition of \cref{def: embedding}. 
This allows both the information-theoretic and spacetime causal orderings (given via $\C{G}_{\C{F}}$ and $\C{T}_{\C{F}}$ respectively) to a priori be frame-dependent, while only keeping the identity of the operationally relevant variables frame-independent.\\


In a spacetime embedded protocol, both notions of causal ordering (information theoretic and spacetime) apply to the same set of observed variables. At this stage, no constraints have been imposed on how these two orderings must interact.
Previous frameworks like \cite{vilasini2022general} 
formalise relativistic causal principles as  precise compatibility conditions between the information-theoretic and spatiotemporal notions of causal orders. For example, \cite{vilasini2022general} considered a variety of such compatibility conditions to formalise the distinct relativistic causality principles such as no-superluminal causation, no-superluminal signalling, and the absence of backward-in-time influence. In the present work, however, we will only require a weaker no-superluminal-signalling condition tailored to protocols involving two observed random variables (the specific communication protocol we consider will be introduced later). We will not review the full compatibility conditions of the prior frameworks, but only present below the relevant weaker definition needed in our work.

 \begin{definition}[No superluminal signalling between two observed variables]
 \label{def: NSS_2var}
Consider a spacetime embedded protocol according to \cref{def: protocol_spacetime} and two observed variables therein $X,Y\in \C{X}_\mathrm{obs}$. 
We say that the protocol satisfies no superluminal signalling from $X$ to $Y$ if the following holds for each frame $\C{F}$ where $X$ is free in the causal graph $\C{G}_{\C{F}}$ of the causal model associated with $\C{F}$,

 $$X \text{ signals to } Y  \text{ (cf. \cref{def: signalling})} \quad\Rightarrow\quad \C{E}_{\C{F}}(X)\xrightarrow{\C{F}} \C{E}_{\C{F}}(Y) \text{ (cf. \cref{def: regions})},$$
 In other words, in each frame $\C{F}$, whenever  $X$ is free and is correlated to $Y$ via the observed distribution $P_{\C{F}}$, the spacetime region (in the spacetime poset $\C{T}_{\C{F}}$) assigned to $X$ must be ordered prior to the region assigned to $Y$. 
 \end{definition}

 The framework therefore separates three questions: whether \(X\) and \(B\) are correlated, whether their causal relation is preserved across frames, and whether their embedded regions remain compatibly ordered in each frame. The no-go theorem shows that these requirements cannot all be maintained under the assumptions introduced next.

\section{No-go theorem}\label{sec6}
The discussion in \cref{sec5} highlights a tension between the communication protocol and the interpretation proposed by Dragan and Ekert. Dragan and Ekert motivate spacetime uncertainty as a possible means of reconciling superluminal reference frames with the absence of superluminal signalling. However, neither the required form of this uncertainty nor the assumptions under which it would suffice have been formulated operationally. We now apply the framework of \cref{sec4} to the communication protocol shown in \cref{cp} and introduced in \cref{sec5} and identify a set of assumptions that cannot be satisfied simultaneously.\\

 We will now state the assumptions that were implicitly or explicitly considered while discussing the communication protocol. Recall the spacetime-embedded protocol of \cref{def: protocol_spacetime} and the communication protocol depicted in \cref{cp}. In the present setting, the protocol contains two observed variables, Alice’s input $X$ and Bob’s output $B$. In each reference frame $\C{F}$, these variables are embedded into the spacetime regions $\C{E}_{\C{F}}(X)$ and $\C{E}_{\C{F}}(B)$, respectively. The assumptions below concern this fixed pair of variables and their associated regions across reference frames.\\

A0. \textbf{The observed variables $X$ embedded in $R_X$ and $B$ embedded in $R_B$ are correlated, $P_{\C{F}}(X, B) \neq P_\C{F}(X) P_{\C{F}}(B)$ in at least one reference frame $\C{F}$}

Correlation between $X$ and $B$ in frame $\C{F}$ does not, by itself, imply signalling in either direction, since the correlation could instead arise from a common cause. However, if $X$ is also free in $\C{F}$, then the correlation implies that $X$ signals to $B$ in that frame, yielding an operational communication protocol \cite{vilasini2022general}.\\

A1. \textbf{Existence of non-order-preserving transformations on operational regions}\\

There exists a non-order-preserving reference-frame transformation

\[
\C{S}_{\C{F} \rightarrow \C{F}'}: \C{T}_{\C{F}} \rightarrow \C{T}_{\C{F}'}
\]
such that\[
R'_X=\C{S}_{\C{F} \rightarrow \C{F}'}\C{E}_{\C{F}}(X),\qquad R'_B=\C{S}_{\C{F} \rightarrow \C{F}'}\C{E}_{\C{F}}(B),
\]

where $R'_X, R'_B \in \C{R}(\C{T}_\C{F'})$ and where the transformed regions do not preserve the original ordering. That is, if the causal ordering in the reference frame $\C{F}$ is
\[R_X \xrightarrow{\C{F}} R_B\]
then, in the $\C{F'}$ 
\[
R'_X \cancel{\xrightarrow{\C{F}'}} R'_B.
\]where $R'_X, R'_B \in \C{R}(\C{T}_\C{F})$.\\

An example is provided by the superluminal Lorentz extension in \cref{LB2}. Note, for our assumption, we only introduce a weaker non-order-preserving transformation at the level of regions, unlike the non-order-preserving transformations in terms of spacetime points introduced in the framework (\cref{RF}). \Cref{ordernp} illustrates how a non-order-preserving transformation acts at the level of regions. It is important to note that non-order preservation for a pair of space-time points does not automatically imply non-order preservation for the regions containing them, since the region may extend sufficiently large to retain their ordering after the transformation (we will discuss such situations in \cref{notA1}). 
 \footnote{For simplicity, our diagrams are constructed with superboosts in SpLTs in mind. Such transformations exchange the time and space coordinate. Hence, the frame-dependent regions coincide diagrammatically with those in the original frame and therefore appear unchanged. However, formally, the regions are still frame-dependent and under non-superboost or nonlinear SpTs, their diagrammatic representation will generally differ in more obvious ways between frames. Nevertheless, the argument extends straightforwardly to general SpTs.} Another example is our illustration of the communication protocol \cref{swap2}.

\begin{equation}\label{ordernp}
    \includegraphics[width=1\linewidth]{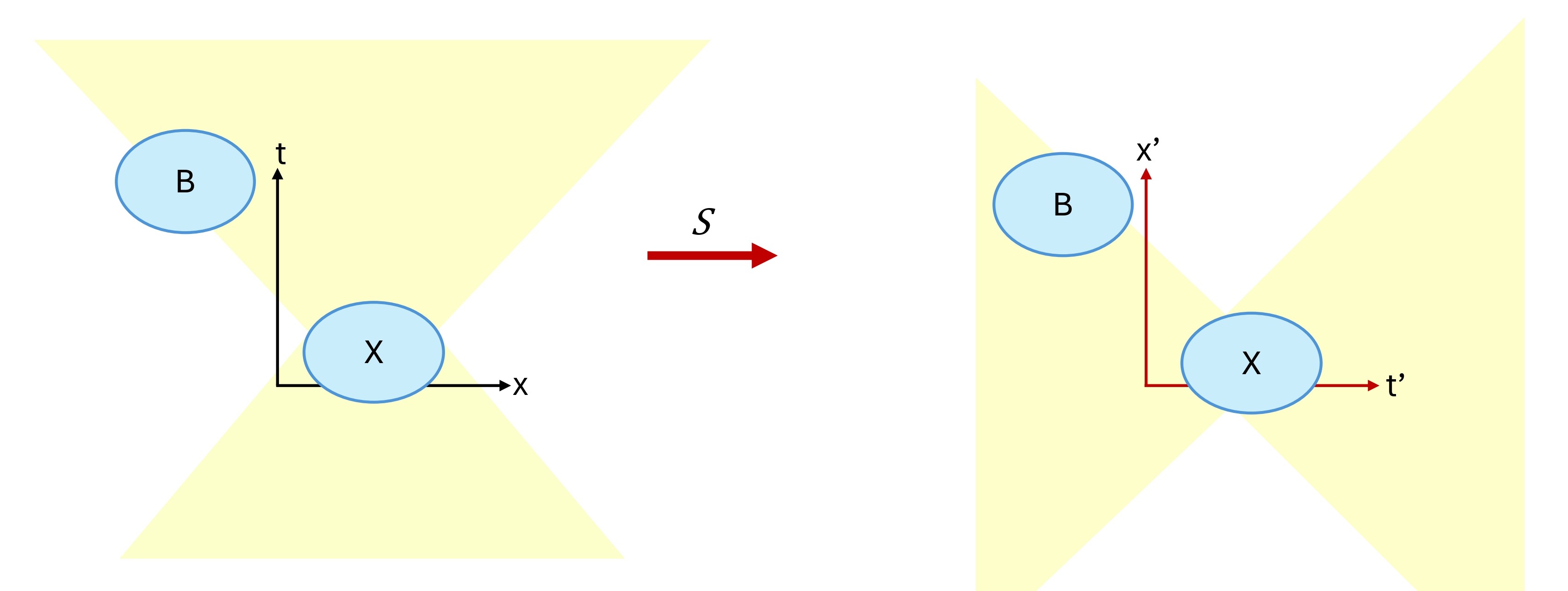}
\end{equation}

A2. \textbf{No superluminal signalling}

 The communication protocol satisfies no superluminal signalling from $X$ to $B$ according to \cref{def: NSS_2var}.\\ 

A3. \textbf{Preserve probabilities under frame transformation}

Observed correlations are independent of the reference frame, i.e., for all reference frames $\C{F}, \C{F}'$\\

$P_\C{F}(B)=P_\C{F'}(B)$; $P_\C{F}(X)=P_\C{F'}(X)$ and $P_\C{F}(B\mid X)=P_\C{F'}(B\mid X)$.\\

In the proposal of \cite{dragan2020quantum}, the requirement that operational statistics agree across inertial frames is a direct expression of the principle of relativity. This requirement plays an essential role in their reasoning; a frame-dependent distinction between predictable and unpredictable emissions would make the frames operationally distinguishable, contrary to the relativity principle.\\

 A4. \textbf{$X$ is free in all reference frames} 

We understand the variable $X$ as an input setting that Alice can freely choose. Hence, it should be \emph{free} in $\C{F}$ and with respect to the causal graph $\C{G}_{\C{F}}$ according to \cref{def: freevar} and it should remain so in any other reference fame $\C{F'}$ with respect to the causal graph $\C{G}_{\C{F'}}$.

\newtheorem*{theorem*}{Theorem}
\begin{theorem*}
 There exists no physical theory that simultaneously satisfies assumptions A0-A4.

\end{theorem*}

\begin{proof}

By assumption A0, there exists a reference frame $\C{F}$ and observed variables $X, B$ such that 
\[
P_\C{F}(X,B)\neq P_\C{F}(X)P_\C{F}(B).
\]
By A4, $X$ is free in $\C{F}$, hence, by \cref{def: signalling}, $X$ signals to $B$ in $\C{F}$. 

Assumption A2 (equivalently, \cref{def: NSS_2var}), then requires
\[
\C{E}_\C{F}(X) {\xrightarrow{\C{F}}} \C{E}_\C{F}(B).
\]
By assumption A1, there exists a non-order-preserving reference frame transformation $\C{S}_{\C{F} \rightarrow \C{F}'}: \C{T}_{\C{F}} \rightarrow \C{T}_{\C{F}'}$ such that the transformed operational
regions
\[
R_X' = \C{E}_\C{F'}(X)= \C{S}_{\C{F} \rightarrow \C{F}'}(R_X), \qquad R_B'=\C{E}_\C{F'}(B) = \C{S}_{\C{F} \rightarrow \C{F}'}(R_B)
\]
where $R'_X, R'_B \in \C{R}(\C{T}_{\C{F}'})$ do not preserve their original causal ordering,
\[
\C{E}_\C{F'}(X)  \cancel{\xrightarrow{\C{F'}}} \C{E}_\C{F'}(B).
\]

By assumption A3, $X$ and $B$ remain correlated in $\C{F}'$, and by assumption A4, $X$ remains free. Thus \cref{def: signalling} implies $X$ signals to $B$ in $\C{F}'$. Assumption A2 therefore requires in frame $\C{F}$,
\[
\C{E}_\C{F'}(X) {\xrightarrow{\C{F'}}} \C{E}_\C{F'}(B).
\]
This contradicts the non-order preservation required by A1,
\[
\C{E}_\C{F'}(X) \cancel{\xrightarrow{\C{F'}}} \C{E}_\C{F'}(B).
\]

\begin{equation}
    \includegraphics[width=0.7\linewidth]{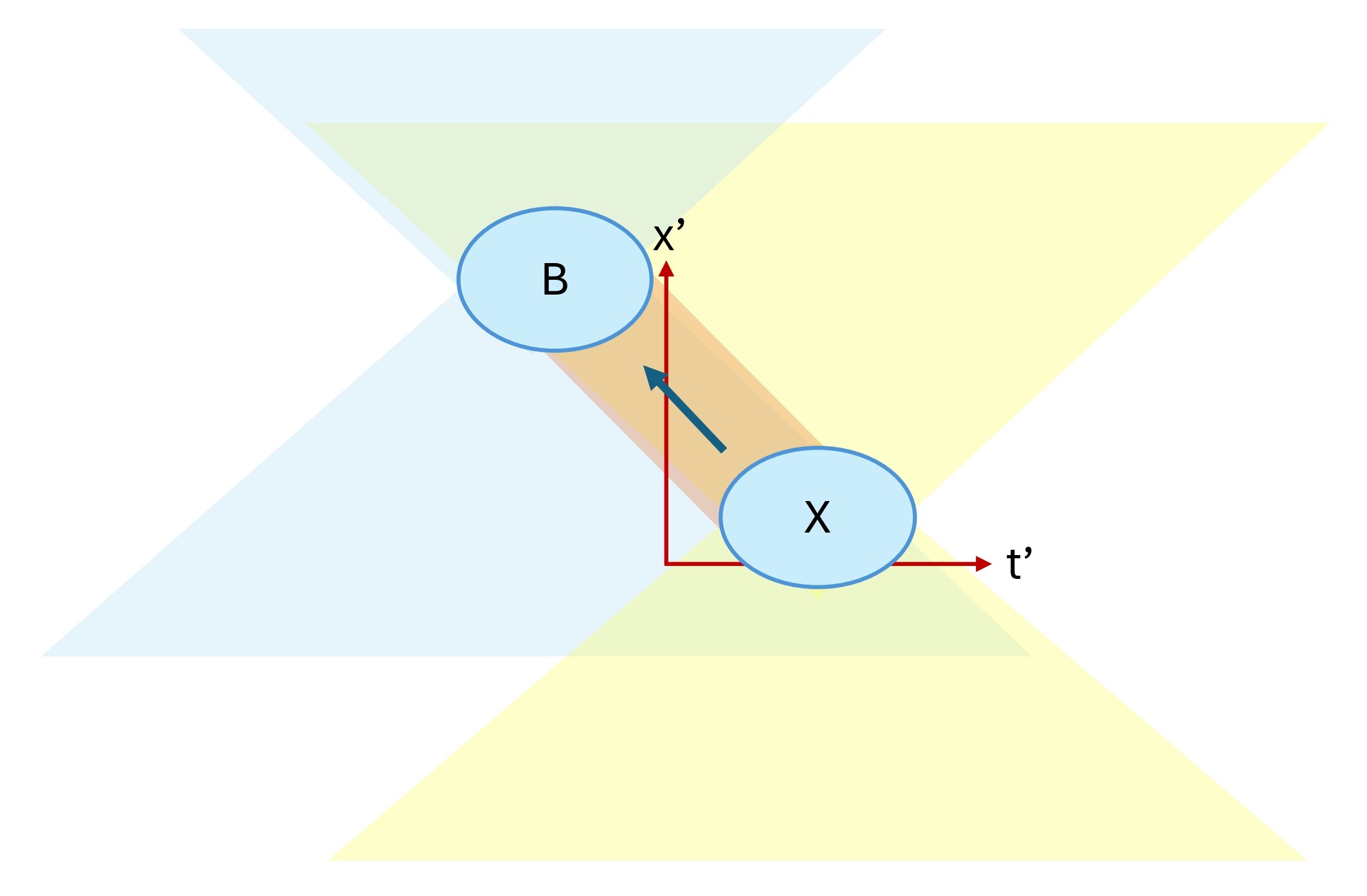}
    \label{signallingback}
\end{equation}
Therefore assumptions A0--A4 cannot all hold simultaneously. \end{proof} 

The theorem does not identify a unique resolution. It instead shows that any model admitting the transformation in A1 must relinquish at least one of the remaining assumptions or restrict its domain so that A1 does not apply. In \cref{sec7} we examine the alternatives and their implications for the uncertainty proposed in \cite{dragan2020quantum}.

 \section{Relaxing the assumptions and implications on D-E uncertainty}\label{sec7}

The no-go theorem established in the previous section demonstrates that assumptions A0–A4 cannot all be satisfied simultaneously. Consequently, any framework admitting superluminal reference-frame transformations must abandon at least one of these assumptions. The theorem therefore does not single out a unique resolution; rather, it delineates the range of logically consistent alternatives. In this section, we will discuss the consequences of violating each assumption. In particular, we discuss why giving up these assumptions is undesirable from a general philosophical point of view of accepting a theory with SpTs, and also why it is undesirable from the perspective of \cite{dragan2020quantum}.\\ 

\subsection{$\lnot$ A0: Correlation (and hence protocol) does not exist}
The first possibility is to reject assumption A0 itself. This amounts to rejecting, in every reference frame, correlations between observed variables of the kind routinely realised experimentally. This would, however, be rather absurd as the existence of correlations is necessary for reasonable physical experiments. Rejecting A0 therefore avoids the no-go result only by excluding the operational phenomenon that \cite{dragan2020quantum} introduced to describe.
Giving this up would mean such correlations do not exist in a world that allows non-order-preserving transformations, for example a world accommodating superluminal reference frames \cite{parker1969, dragan2020quantum}. In this case, D-E uncertainty is unnecessary because the correlation that gives rise to the superluminal-signalling problem is absent.

\subsection{$\lnot$ A1: Frame transformations are always order preserving on operational regions}\label{notA1}

 Assumption A1 requires the existence of at least one reference frame transformation $\C{S}_{\C{F}\rightarrow\C{F'}}$ that fails to preserve the causal ordering of the operational regions associated with the protocol. There are two ways to contradict this.\\

 The first option excludes SpTs and recovers the conventional special-relativistic setting, in which Lorentz transformations preserve the light-cone order. No D-E uncertainty is then required.
 While this is a completely viable option, it goes against the premise we set out to study, namely, which constraints SpTs impose on physical theories. We therefore do not consider this option further. 

\begin{equation}
    \includegraphics[width=1\linewidth]{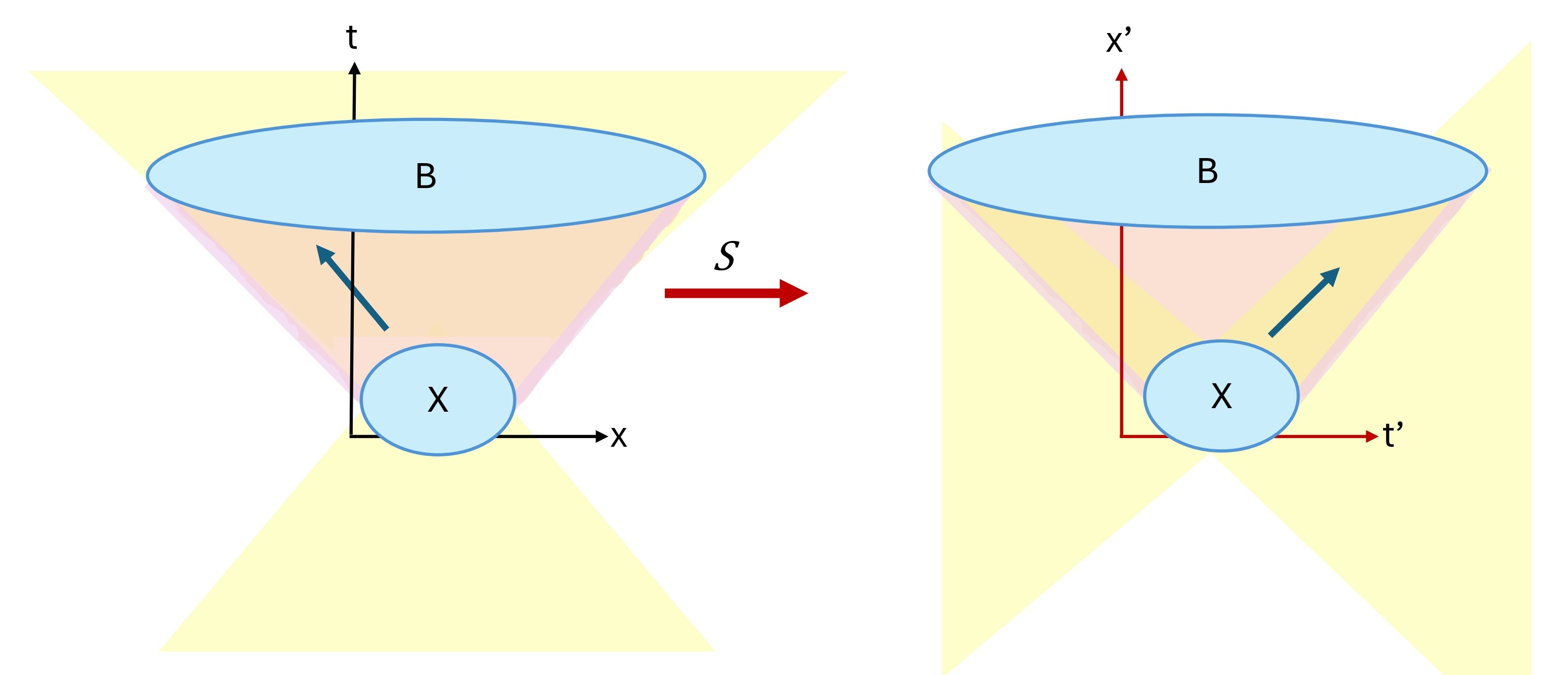}\label{del}
\end{equation}

Second, one may allow SpTs to change the ordering of spacetime points while requiring the operational regions $R_X$ and $R_B$ to remain order preserving. This immediately implies that these operational regions cannot be points. These regions must then have sufficient spatiotemporal extent so that a variable embedded in such a region has sufficient uncertainty in its location. This would mean, for every admissible transformation, in every frame $\C{F}'$, $\C{E}_\C{F}'(X)\Rprec{\C{F}'} \C{E}_\C{F}'(B)$. For example, the support of Bob's variable $B$ may extend over a region large enough that at least one spacetime point in $R_B$ is always in the causal future of $X$, see \cref{del}. The communication protocol could then be embedded consistently in every frame without superluminal signalling.\\
 
Conceptually, this resolution is closely related to the proposal of \cite{dragan2020quantum}, according to which agents cannot control the exact time at which a superluminal particle is emitted and therefore cannot use the process to transmit information superluminally. In our framework, this lack of control is represented by sufficiently extended spacetime uncertainty regions that prevent the protocol from defining operational superluminal signalling.\\  

The issue, however, is that the uncertainty becomes too large. Consider a setup as in \cref{swap2} under a superluminal boost. For the operational regions associated with $X$ and $B$ to retain their ordering in the superboosted reference frame, $R_B$ must have significant spatial uncertainty, scaling with the temporal distance between $R_X$ and $R_B$. This is illustrated in \cref{del}. The ``operational'' regions then lose essentially any operational meaning. For a sufficiently macroscopic, i.e., classical Bob, we would expect the spatial uncertainty we have about the location of his guess $B$ 
to be at most approximately equal to the spatial extension of his lab, so on the order of 1-10 metres. However, already for a temporal distance of 1 second between Alice and Bob, the spatial uncertainty to violate A1 is on the order of 1 light-second, or approximately 300,000 kilometres. Thus, even for an ordinary subluminal communication protocol, avoiding A1 through D-E uncertainty requires an operational region many orders of magnitude larger than what is experimentally feasible.

\subsection{$\lnot$ A2: Superluminal signalling}
Another route to evading the no-go theorem is to reject A2 and permit operational superluminal signalling.  The apparent influence from Alice's input $X$ to Bob's output $B$ may then propagate outside the future lightcone, or even appear to propagate backwards in time for some observers, as illustrated in \cref{signallingback}. 
 The cost of this resolution is not merely the abandonment of a single foundational principle, but rather of two. In addition to violating assumption A2, it undermines the central motivation of relativity by introducing preferred reference frames in which superluminal signalling is permitted.\\  
 
We further note that the entire premise of introducing D-E uncertainty was to prevent superluminal transformations and tachyons from enabling superluminal communication. If superluminal signalling is instead accepted, we no longer need D-E uncertainty nor must we abandon local determinism for the reason proposed in \cite{dragan2020quantum}. Rejecting A2 therefore avoids the no-go result by relinquishing the very constraint that D-E uncertainty was introduced to preserve.

\subsubsection{Retrocausal explanation}

So far, we have focused on the observed variables $X$ and $B$. The same signalling correlation may, however, admit a causal explanation involving an unobserved variable $\lambda \in \mathcal{X}_{\mathrm{nobs}}$ that mediates their connection. In the classical setting considered here, $\lambda$ may be interpreted as a hidden variable.\\

The nonlocal hidden variable may then have \(X\) as a parent and may, in turn, be connected to \(B\) by directed edges. Whether these directed edges are retrocausal depends on how the nonlocal hidden variables are embedded relative to the chosen reference frame. In both \cref{linelambda}, \cref{linelambdafuture}, \(X\) remains free and correlated with \(B\), and therefore signals to \(B\); however, on the right-hand of the diagrams \(B\) lies in the causal past of \(X\), so we violate A2.
\begin{equation}
		\centering
		\includegraphics[width= 1 \textwidth]{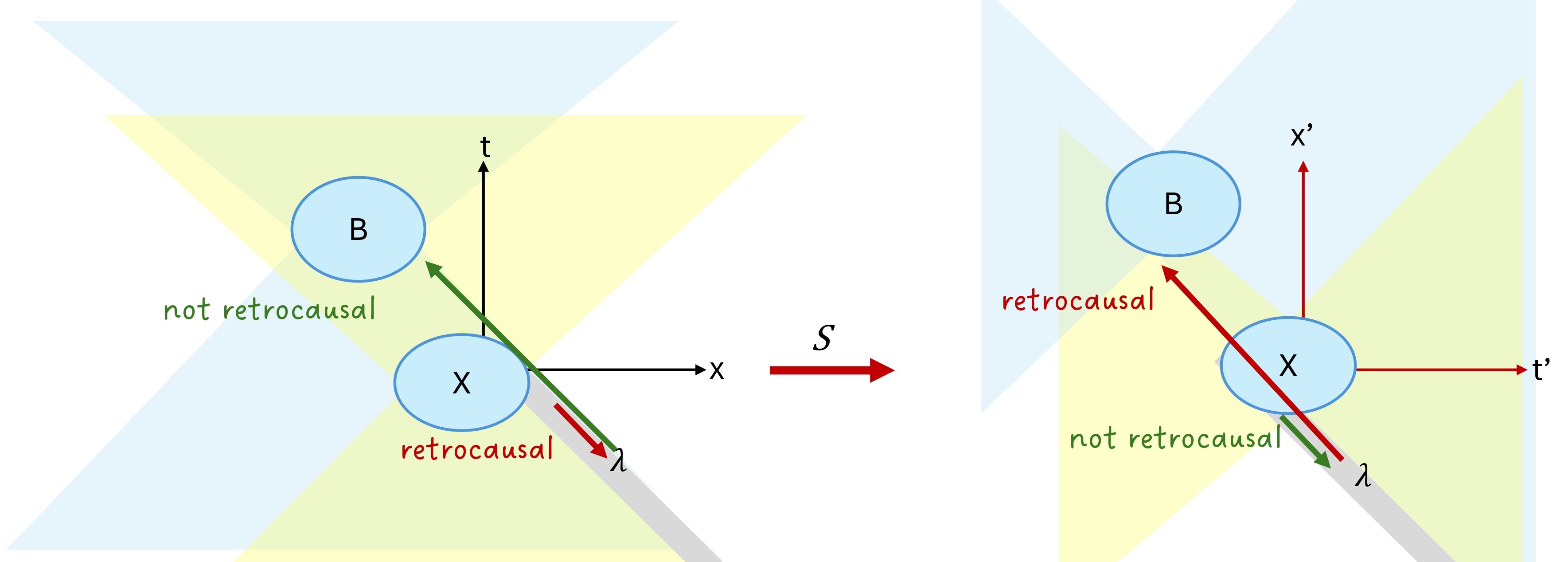}	
		\label{linelambda}        
	\end{equation}

\begin{equation}
		\centering
		\includegraphics[width= 1 \textwidth]{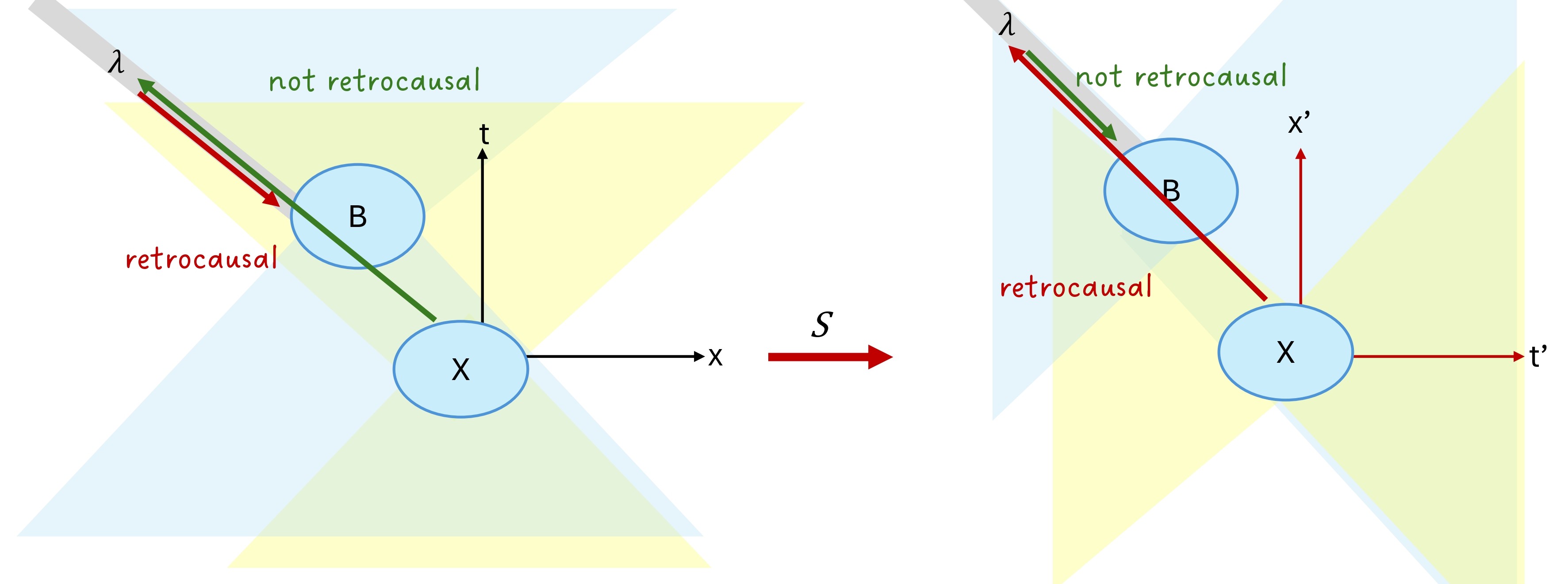}	
		\label{linelambdafuture}        
	\end{equation}

\subsection{$\lnot$ A3: Frame-Dependent Operational Statistics}

The no-go theorem can also be avoided by abandoning A3 at the expense of frame-dependent and observer dependent operational statistics, and violating the principle of relativity. Thus, the cost is relinquishing a strong form of observer independence: distinct reference frames would no longer agree on the operational probabilities associated with the same communication protocol. This suggests that the choice of the reference frame has observable consequences. In particular, if a non-order-preserving superluminal transformation maps the protocol to a frame $\C{F}'$ in which $X$ is not embedded in the causal past of $B$, the transformed statistics may satisfy $P_\C{F}(B|X) \neq P_\C{F'}(B|X)$. The correlation between $X$ and $B$ could 
then disappear in $\C{F}'$, allowing A2 and A4 to remain valid despite the change in operational causal order.\\ 

We note that existing frameworks like quantum reference frames \cite{Ahmad2022QuantumRelativitySubsystems, Hoehn2023SubsystemRelativity, CastroRuiz2025RelativeSubsystems, Cepollaro2025SubsystemCoherence, sen2026superluminal} which are motivated by the conceptually radical idea that  we can ``jump" into the perspective of superposed reference frames, nevertheless satisfy assumption A3. In such settings, even though certain quantum properties like entanglement and superposition are frame-dependent, the observable statistics or the probabilities that one would get after suitable measurements are still invariant across all reference frames.\\  

Finally, this option weakens the motivation for the specific resolution advocated in \cite{dragan2020quantum}, thereby violating principle of relativity. If observed correlations are reference frame-dependent and satisfy no-superluminal signalling, then no additional D-E uncertainty is needed. The variables may remain embedded at sharply localised spacetime points, while their correlation changes whenever a transformation alters the ordering of those points.

\subsection{$\lnot$ A4: Alice's variable $X$ is not free in at least in some reference frames}
The final option to evade the no-go theorem is to abandon A4, i.e., the existence of free variable $X$ across all possible references frames. This can be done in two ways: by switching which of $X$ and $B$ is free (\cref{swap}) or by giving up the notion of free variables altogether (\cref{supd}). We examine these possibilities in turn.

\subsubsection{The sender-receiver swap}\label{swap}
A4 is violated if free-ness is a reference-frame-dependent property; in some reference frames $X$ is a free variable signalling to $B$ and in others $B$ is a free variable signalling to $X$, with neither of $X$ or $B$ being fundamentally free in a frame-independent way (see \cref{sd}). This blocks the proof by preventing the correlation in the transformed frame from being interpreted as signalling from $X$ to $B$.
\begin{equation}
    \includegraphics[width=1\linewidth]{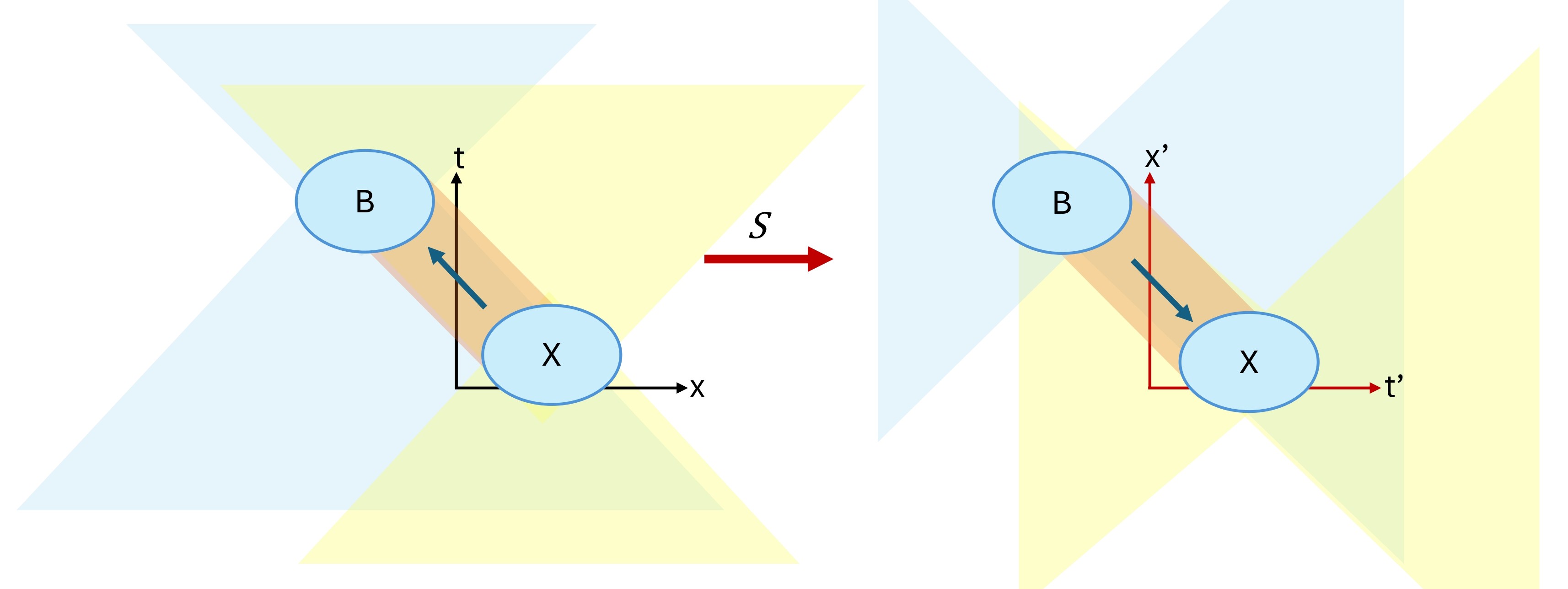}\label{sd}
\end{equation}

A similar approach appears in treatments of the tachyon energy problem, where the sender and receiver are re-identified in different reference frames to avoid apparent energy paradoxes \cite{schwartz2018tachyon, sen2026superluminal}.\\

The sender--receiver swap comes at the cost of making the operational distinction between inputs and outputs frame-dependent. This raises a more fundamental question about the meaning of a free variable. If a variable is free in one reference frame but not in another, what does it mean to say that $X$ represents Alice’s free choice from her perspective, if, from Charlie’s perspective, the same variable $X$ is instead a consequence of Bob’s actions? In an ordinary communication protocol, the distinction between an input and an output is operationally clear. Alice chooses a bit, implements that choice using her device, and Bob subsequently records the corresponding output. If Charlie describes the same experiment in a reference frame in which the causal roles are reversed, how should he interpret Alice’s choice and Bob’s measurement?\\

Swapping the sender and receiver may seem more satisfactory than any of the previous options from the perspective of \cite{dragan2020quantum}, however is not sufficient in every configuration. For example, the one depicted in \cref{swap2} where reversing the direction of signalling by itself do not restore the ordering required by A2. We note that the D-E uncertainty required is hence proportional to the difference in communication speed and speed of light (since, timelike communication becomes spacelike under a SpT). Hence, one of A1, A2 or A3 would still need to be dropped.\\

While the previous example shows that the sender-receiver swap is insufficient in some scenarios, it is at the same time too efficient in others, removing any need to have uncertainty. Consider a situation where the communication protocol is between observers localised in lightlike separated spacetime points $r_X, r_B$. This could be understood as communication carried out using photons. In our abstract framework, these lightlike separated points correspond more generally to points which can have their ordering flipped, but never become unordered under any reference frame transformation. Superluminal signalling is thus avoided by simply swapping which party is the sender and which the receiver depending on the ordering.  
Hence, the sender–receiver swap does not require any uncertainty in the spacetime localisation of photons, and the argument of \cite{dragan2020quantum} does not explain why uncertainty should arise in the lightlike case. A proposed theory of superluminal observers aimed at combining relativity and quantum mechanics in the way proposed by \cite{dragan2020quantum} would have to include an additional mechanism that explains the uncertainty in the lightlike case, and only in that case so as to not render superluminal observers unnecessary. Hence, superluminal observers in the way advocated by \cite{dragan2020quantum} \textit{on its own} are sufficient at most to explain the uncertainty observed when it comes to massive particles. An additional mechanism would need to be identified for massless particles. If this can be done, the question then is whether this additional mechanism is not actually sufficient and it is not necessary to invoke D-E uncertainty anymore.\\

\subsubsection{Superdeterministic explanation}\label{supd}
In the superdeterministic explanation, the apparent sender–receiver structure of the protocol need not correspond to a genuine causal influence from $X$ to $B$. Instead, both variables may be jointly determined by a hidden variable, such that the correlations required by the protocol arise from pre-established constraints rather than dynamical signalling. Indeed, the cost is the absence of free variables, and in terms of the communication protol Alice's \emph{free} will to open (or not) open the gate. Here, we  give up the requirement of \emph{free} $X$  
such that for a non-local hidden variable $\lambda$, the probabilities $P_{\C{F}}(X \mid \lambda) \neq P_{\C{F}}(X)$. One possible way to diagrammatically represent is \cref{commoncause}. Note, the spread of $\lambda$ is not required to be contained in some region separate from $B, X$, it could also be also spread over the regions of $X,B$.

\begin{equation}
		\centering
		\includegraphics[width= 1 \textwidth]{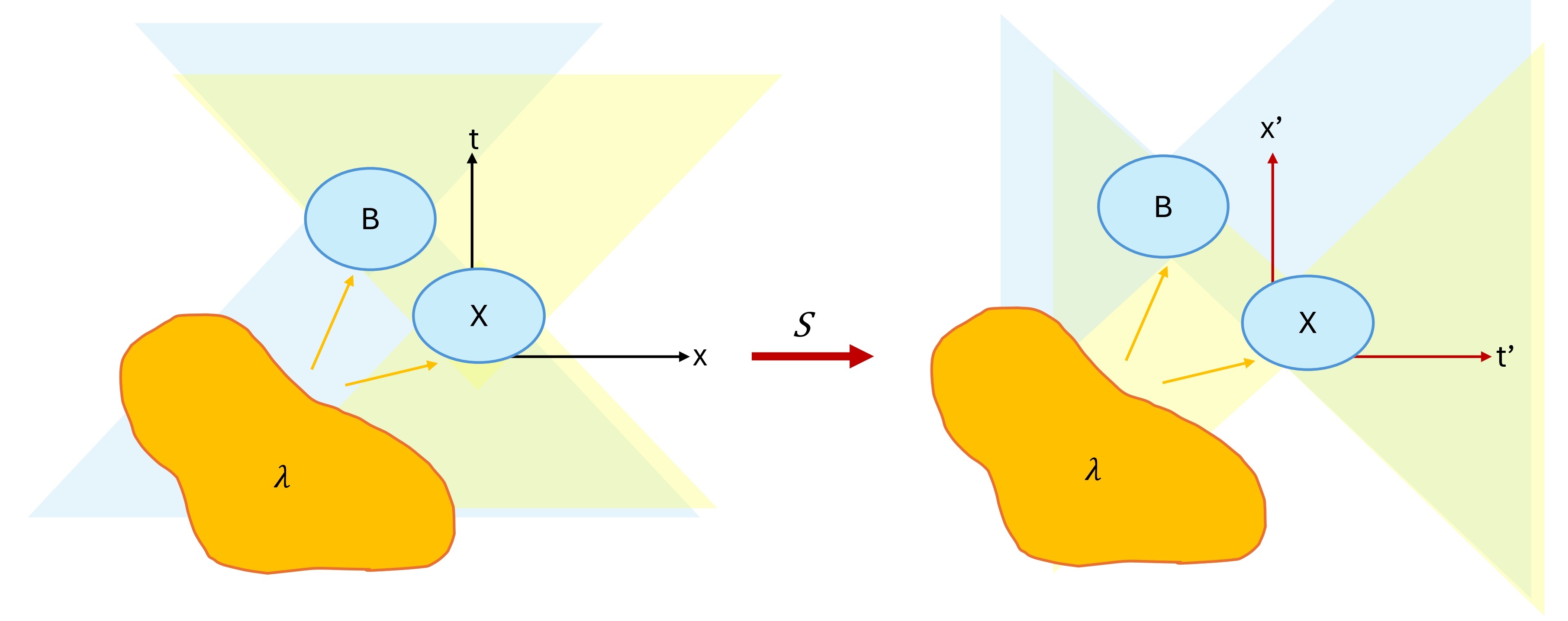}	
		\label{commoncause}        
	\end{equation}

The superdeterministic explanation comes in two flavours, one where there is still an edge in the causal graph from $X$ to $B$ and one where there is no such edge. The latter is undesirable from the point of view of \cite{dragan2020quantum}. If there is no causal influence between observed variables in the first place, then it is not necessary to introduce any spatiotemporal uncertainty to remove superluminal influences. Of course, it could simply be that this is the way the world is as uncertainty is not excluded. But superluminal observers no longer necessarily force a rejection of local determinism, they merely accommodate this rejection. In fact, if this path is adopted, the Dragan–Ekert argument and the requirement of D-E uncertainty no longer apply, since there is no causal influence between observed variables in the first place, it is not necessary to introduce any spatiotemporal uncertainty to remove superluminal influences.\\ 

The superdeterministic explanation in which an edge from $X$ to $B$ remains may then appear to be consistent with \cite{dragan2020quantum}. However, we argue that this possibility is likely an artefact of the relatively weak way in which we formalise the absence of superluminal signalling in assumption A2, rather than a genuine limitation of the no-go argument. In particular, A2 only constrains signalling from free observed variables. Indeed, our framework does not even define when correlations involving two non-free observed variables constitute signalling. It is natural to require that non-free variables likewise cannot signal superluminally to other variables, once a well-defined notion of interventions and signalling for such variables is defined. Under such an extension, if the edge $X \xrightarrow{\C{F}}B$ is taken to imply signalling from $X$ to $B$ in reference frame $\C{F}$ (e.g., as discerned by intervening on the non-free $X$ and verifying correlation with $B$ in the post-intervention model \cite{vilasini2022general}), then this signalling must be superluminal whenever A1 holds. We therefore expect that superdeterministic explanations with an edge from $X$ to $B$ 
can be ruled out by extending our framework with a suitable notion of signalling between non-free variables, for example using the affects-relation framework of \cite{vilasini2022general}.\\

For ease of comparison, \Cref{tab:relaxing-assumptions} summarises how relaxing each assumption
blocks the proof, the corresponding conceptual cost, and the resulting
implications for D-E uncertainty. The entries identify logical routes
around the no-go result; they should not be interpreted as establishing
the consistency or physical viability of the resulting theories.
\begin{table*}[t]
\centering \footnotesize
\renewcommand{\arraystretch}{1.25}
\caption{Possible ways of avoiding the no-go result. Relaxing an
assumption removes the contradiction but does not, by itself, establish
that the resulting theory is physically consistent.}
\label{tab:relaxing-assumptions}

\setlength{\fboxsep}{5pt}   
\setlength{\fboxrule}{0.5pt} 

\fbox{%
\begin{minipage}{
  \dimexpr\linewidth-2\fboxsep-2\fboxrule\relax
}

\begin{tabularx}{\textwidth}{
    >{\raggedright\arraybackslash}p{2cm}
    >{\raggedright\arraybackslash}X
    >{\raggedright\arraybackslash}X
    >{\raggedright\arraybackslash}X}
\textbf{Relaxed assumption}
&
\textbf{How the no-go theorem is avoided}
&
\textbf{Main cost}
&
\textbf{Implication for D-E uncertainty}
\\
\addlinespace[5pt]
A0
&
The operational correlation between \(X\) and \(B\) is absent.
&
The theory does not allow correlations, making reasonable experiments impossible.
&
D-E uncertainty is unnecessary because the correlation giving rise to the signalling problem has been removed.
\\
\addlinespace[5pt]
A1
&
Exclude non-order-preserving transformations or require operational regions to remain ordered.
&
The theory either excludes SpTs or requires sufficiently extended operational regions.
&
D–E uncertainty must be so large that even ordinary subluminal communication requires Bob’s operational region to extend far beyond normal experimental scales. 
\\
\addlinespace[5pt]
A2
&
Signalling is permitted even when the operational regions are not
causally ordered.
&
Rejecting a foundational principle and potentially violating principle of relativity.
&
D-E uncertainty is no longer required to prevent superluminal
signalling, since such signalling is accepted.
\\
\addlinespace[5pt]
A3
&
The correlation present in one frame need not persist in another.
&
Frame-dependent operational statistics and violation of principle of relativity.
&
D-E uncertainty may be unnecessary because the transformed protocol
need not retain the original correlation.
\\
\addlinespace[5pt]
A4
&
Which variable is free depends on the frame or free variables do not exist
&
Input-output roles become frame dependent or lack of free choice, i.e., superdeterminism.
&
D-E uncertainty is unnecessary for lightlike communication under a sender--receiver swap as well as in a superdeterministic account, where the variables are predetermined.

\\

\end{tabularx}
\end{minipage}%
}
\end{table*}

\section{Conclusion}\label{sec8}

In this work, we investigated the causal implications of extending the principle of relativity to include superluminal reference-frame transformations. The proposal in \cite{dragan2020quantum} shows that tachyons are incompatible with a locally deterministic theory, but it does not establish that a theory containing tachyons becomes consistent once local determinism is abandoned. If we recall Dragan-Ekert's argument discussed in \cref{sec5}, their reasoning begins by assuming no local description (the past lightcone not having information about emission and absorption in the presence of superluminal reference frames). They then claim that it is necessary to abandon local- \emph{determinism} in the presence of superluminal particles, given that there is no-superluminal signalling since we cannot \emph{control} the emission and absorption of tachyons or superluminal particles. To us, this reasoning appears circular and lacks the formal definitions and framework needed to clarify the relations between locality, determinism, operational control, and signalling.\\

Dragan and Ekert appeal to a fundamental form of uncertainty (when they abandon local-determinism), which we call D-E uncertainty and formalise by embedding relevant events in spacetime regions. In this work, we developed a framework in which both the spacetime  and the causal model are allowed to depend on the reference frame and communication protocols such as the one considered in \cite{dragan2020quantum} can be analysed explicitly.\\ 

Within this framework, we proved a no-go theorem demonstrating that the existence of correlations in any reference frame (A0), the existence of superluminal transformations which can change the causal ordering of operational regions (A1), the absence of superluminal signalling (A2), invariance of operational probabilities under reference-frame transformations (A3), and invariance of free variables across reference frames (A4) are inconsistent. The theorem therefore provides a systematic characterisation of the conceptual trade-offs required by theories admitting superluminal reference frames. The theorem identifies this incompatibility but does not determine which assumption should be abandoned or whether every resulting alternative defines a complete and physically satisfactory theory.\\ 

Our result also clarifies the limitations of arguments in \cite{dragan2020quantum} based solely on objective uncertainty. As discussed in \cref{notA1}, D-E uncertainty can evade the no-go theorem only if it sufficiently large that the operational regions remain order-preserving under all admissible reference-frame transformations. This is a substantially stronger requirement than merely postulating indeterminacy, and may imply that agents cannot, even in principle, localise or control communication events and even standard classical scales become unachievable. Thus, the first issue is not whether uncertainty exists, but whether it is sufficiently extensive to preserve the causal ordering of operational regions in every reference frame. Second, even if one accepts that D-E uncertainty may provide one possible way of addressing certain causal paradoxes associated with SpTs, it is neither shown to be sufficient on its own nor uniquely required by the existence of such transformations. Relaxing other assumptions opens alternative routes around the no-go result, which do not necessarily require D-E uncertainty, although their physical viability must be assessed separately. Finally, we also note, D-E uncertainty is specifically a form of spatiotemporal uncertainty concerning the localisation of events in space and time, whereas quantum uncertainty is more general and can apply to other incompatible observables, for example position and momentum.\\

Our result does not rule out all theories admitting superluminal reference frames; rather, it constrains theories that also seek to preserve the operational principles represented by A0–A4. Our work also demonstrates how causal modelling provides a useful framework for analysing extensions of relativity. By making explicit the assumptions underlying different proposals, it becomes possible to identify precisely where inconsistencies arise and to compare alternative approaches on equal footing. We hope that this framework will prove useful in future investigations of superluminal reference frames, relativistic causality, and more general extensions of spacetime symmetry.\\

\section{Discussion}\label{sec9}

The significance of the present result extends beyond the particular communication protocol considered here. Rather than proposing a specific ontology for superluminal reference frames, our theorem identifies the minimal set of assumptions that become mutually incompatible once non-order-preserving reference-frame transformations are admitted.\\

No-go theorems are particularly important when extending physical theories into regimes where direct experimental tests are difficult or currently impossible, since they help identify which combinations of assumptions remain logically consistent and which must be abandoned even in the absence of immediate empirical verification. Even if one does not see the value in a theory of superluminal extension, it is still interesting to see what a theory involving them brings to the table. Additionally, it lays the groundwork for refining and redefining the fundamental assumptions of successful physical theories, such as special relativity. The present no-go theorem complements the earlier no-go result by one of us (AS) in \cite{Sen2026Indeterminism}. 
The first no-go theorem \cite{Sen2026Indeterminism} talks about the evolution and transformation of a finite informational causal model in the presence of non-order-preserving superluminal transformations, and examines the conditions under which indeterminism can arise. Here, by contrast, we begin with an operational correlation between two observed variables in a communication protocol and ask whether introducing D-E uncertainty is sufficient to prevent that correlation from becoming a superluminal signal in another reference frame.\\

Our theorem highlights an intimate connection between relativistic symmetry and causal structure. If superluminal reference frame transformations fail to preserve the ordering of operational regions, then maintaining the remaining operational principles inevitably requires modifying some aspect of our conventional understanding of causality, observer independence, or free choice. The theorem does not select a preferred modification, nor does it assume that superluminal reference frames are physically realizable. It does not determine which modification Nature (if one believes superluminal observers should be physical) or toy theories with superluminal observers (if one believes otherwise) adopts, but it sharply constrains the possibilities.\\

One might also be tempted to compare this with Sorkin’s paradox in relativistic quantum field theory \cite{sorkin1993impossible}. However, Sorkin’s construction requires three agents and relies on a theory-dependent quantum (field) theoretic description of measurements to observe signalling correlations between spacelike separated regions.  A similar construction in quantum information was done in \cite{PhysRevD.24.359} by Aharonov and Albert. Also note, there are no superluminal frames there, only superluminal signalling via measurement protocols. 
In contrast, our no-go theorem is formulated independently of any particular physical theory or measurement prescription and requires only two parties, demonstrating a more fundamental obstruction.\\ 

Finally, recent work has extended SpTs to quantum reference-frame settings \cite{sen2026superluminal}. While the present framework already accommodates cyclic information-theoretic causal models, it assumes that the spacetime causal structure relative to each reference frame is described by a definite partial order capturing a definite lightcone structure, even if this is frame-dependent \cite{vilasini2022general, Vilasini2024Embedding}. 
It would therefore be interesting to investigate whether our no-go theorem extends to settings in which the spacetime causal structure itself becomes indefinite, such as in superpositions of geometries or quantum reference frames. Extending the present no-go theorem to such settings may provide further insight into the relationship between relativistic symmetry, causality, and quantum theory.\\

{\bf Acknowledgements.} AS acknowledges support from the IRA Programme, project no. FENG.02.01-IP.05 0006/23, financed by the FENG program 2021-2027, Priority FENG.02, Measure FENG.02.01., with the support of the FNP. MS is supported by National Science Centre, Poland (Preludium project, Classicality and compositionality of general probabilistic theories in infinite dimensions, project no.2024/53/N/ST2/01192 and Opus
project, Categorical Foundations of the Non-Classicality of Nature, project no. 2021/41/B/ST2/03149). VV acknowledges support from the PEPR integrated project EPiQ ANR-22-PETQ-0007 as part of Plan France 2030.

--------------------------------------------------------------------

\bibliography{newbib}
\end{document}